\documentclass{birkjour}
 \newtheorem{thm}{Theorem}[section]
 
 \newtheorem{lem}[thm]{Lemma}
 \newtheorem{prop}[thm]{Proposition}
 \theoremstyle{definition}
 
 \theoremstyle{remark}
 \newtheorem{rem}[thm]{Remark}
 
 \numberwithin{equation}{section}

\begin{document}

%
%
%
%
%
%
%
%
%

\title[Singular selfadjoint perturbations of selfadjoint operators]
 {Singular selfadjoint perturbations of unbounded selfadjoint operators. Reverse approach}

\author[Adamyan]{V.M.Adamyan }

\address{%
 Department  of Theoretical Physics and Astronomy\\
Odessa National I.I. Mechnikov  University\\
65082 Odessa\\ 
Ukraine}

\email{vadamyan@onu.edu.ua}

\thanks{This paper is an expanded version of the talk presented by the author at the conference "Spectral Theory and Applications", Stokholm, Sweden, March 13 – 15, 2016 }


\dedicatory{In memory of Boris Pavlov: brilliant mathematician and fascinating personality}

\begin{abstract}
Let $A$ and $A_{1}$ are unbounded selfadjoint operators in a Hilbert space $\mathcal{H}$. Following \cite{AK} we call $A_{1}$ a  \textit{singular} perturbation of $A$ if $A$ and $A_{1}$ have different domains $\mathcal{D}(A),\mathcal{D}(A_{1})$ but $\mathcal{D}(A)\cap\mathcal{D}(A_{1})$ is dense in $\mathcal{H}$ and $A=A_{1}$ on  $\mathcal{D}(A)\cap\mathcal{D}(A_{1})$. In this note we specify without recourse to the theory of selfadjoint  extensions of symmetric operators the conditions under which a given bounded holomorphic operator function in the open upper and lower half-planes is the resolvent of a singular perturbation $A_{1}$ of a given selfadjoint operator $A$. 

For the special case when $A$ is the standardly defined selfadjoint Laplace operator in $\mathbf{L}_{2}(\mathbf{R}_{3})$ we describe using the M.G. Krein resolvent formula a class of singular perturbations $A_{1}$, which are defined by  special selfadjoint boundary conditions on a finite or spaced apart by bounded from below distances infinite set of points in $\mathbf{R}_{3}$ and also on a bounded segment of straight line embedded into $\mathbf{R}_{3}$ by connecting parameters in the boundary conditions for $A_{1}$ and the independent on $A$  matrix or operator parameter in the Krein formula for the pair $A, A_{1}$.   
\end{abstract}

\maketitle
\section{Introduction}
The so called solvable models associated with zero-radius potentials \cite{AGHH} and more general singular perturbations has come to the foreground in the late oeuvre of Boris Pavlov. He and his numerous disciples and followers enriched these models and significantly expanded the boundaries of their applications, endowing the involved point potentials and singular perturbations with internal structures. The results pertaining to the initial stages of the relevant studies can be found in the review \cite{Pav} and subsequent monograph \cite{AK}. 
Recall that Schrödinger operators with potentials of zero radius appeared in physical applications more than 80 years ago (historical references and comments can be found in the well-known books \cite{DeOst},\cite{AGHH},\cite{AK}). However, a clear understanding of the mathematical nature of these objects was achieved much later in \cite{BF}. After the note \cite{BF} the theory of extensions of symmetric operators turned out the main tool for solving the problems of spectral theory and scattering theory for Schrödinger and afterwards for Dirac operators with potentials or analogues of potentials formally given as combinations of Dirac $\delta$-functions. 

As it was traced in \cite{AdP}, \cite{Pav} the solvability of the zero-range potential models and problems for a wide class of singular perturbations of selfadjoint operators lie in the  algebraic simplicity and universality of M.G. Krein resolvent formula for  selfadjoint perturbations of a given selfadjoint operator. It appears that to solve specific problems of spectral and scattering theory for sufficiently wide classes of perturbations of selfadjoint operators the mentioned M.G. Krein formula can be used as the only tool of analysis. 

However, despite the large number of deep and interesting mathematical results on the zero-range potential models and singular perturbations and their effective, elegant and useful physical applications obtained in subsequent years, a profound analysis of related problems with quest for analytically solvable models does not apply to interests of the majority of today's consumers of mathematical physics. Instead, they would prefer to solve  their problems using computer algebra systems and numerical calculations. This paper is an attempt to develop an available to the mass consumer lite theory of singular perturbations of seladjoint operators operating only with that resolvent formula. 
 
In auxiliary Section 2, we recall the necessary and sufficient conditions under which a function on an open set of the complex plane whose values are bounded linear operators in Hilbert space is the resolvent of densely defined closed linear operator, particularly, of selfadjoint operator. We also give here the known description of resolvents for finite-dimensional selfadjoint perturbations of a given selfadjoint operator. 

A short section 3 is devoted to the derivation of the Kerin formula for resolvents of certain classes of singular perturbations of a given selfadjoint operator. Using the approach of M.G. Krein, but not referring to the theory of extensions, we justify a well-known, in our opinion application-friendly parametrizations of this formula.

The first of two obtained version of the Krein formula is illustrated in Section 4 by the example of singular selfadjoint perturbations of the selfadjoint Laplace operator in $\mathbf{L}_{2}(\mathbf{R}_{3})$  that have form of a sum of zero-range potentials spaced apart by bounded from below distances.

The second obtained version of the Krein formula is more suitable for describing singular perturbations of the classical Laplace operator in $\mathbf{L}_{2}(\mathbf{R}_{3})$ whose action is concentrated on one- and two-dimensional manifolds of $\mathbf{R}_{3}.$ This version was illustrated in Section 5 by singular perturbation of the Laplace operator, which is located on a straight-line segment embedded into $\mathbf{R}_{3}$. The role of the parameter in the Krein formula in this case is played by the selfadjoint Sturm-Liouville operator on the given segment. The results of this section can easily be extended to the case when the singular perturbation of the Laplace operator in $\mathbf{L}_{2}(\mathbf{R}_{3})$ is given on a compact quantum graph embedded into $\mathbf{R}_{3}$. In the latter case, we obtain an extension of the proposed in \cite{PF} model for describing the interaction of molecules with the surrounding medium.   

  \section{Reminder of resolvents basic properties}
\begin{thm}\label{initial} 
Let $R(z)$ be a strongly continuous operator function on a non-empty open area $\mathbb{D}$ of complex plane and values of this function are  bounded operators in a Hilbert space $\mathcal{H}$. $R(z)$  is the resolvent of a linear densely defined closed operator $A$ in  $\mathcal{H}$ with the resolvent set $\varrho(A)\supseteq\mathbb{D}$ if and only if 
\begin{itemize}
\item{\begin{equation}\label{first}
\ker R(z)=\ker R(z)^{*}=\{0\};
\end{equation}}
\item{for any $z_{1},z_{2}\in \mathbb{D}$ the Hilbert equality 
\begin{equation}\label{hilb1}
R(z_{1})-R(z_{2})=(z_{1}-z_{2})R(z_{1})R(z_{2})
\end{equation} }
 holds. \end{itemize}

\end{thm} 
\begin{proof} By (\ref{first}) for each $z\in\mathbb{D}$ the linear relation  
\begin{equation}\label{linear}
\left\{\begin{array}{c}
g=R(z)f,\;f\in\mathcal{H} \\
A_{z}g=f+zR(z)f=f+zg
\end{array} \right. 
\end{equation}
defines a linear operator $A_{z}$ with the dense range $R(z)\mathcal{H}$.  

If for some sequence $f_{n}\in\mathcal{H}$ the sequences $g_{n}=R(z)f_{n}$ and $A_{z}g_{n}=f_{n}+zg_{n}$ converge to vectors $g_{\infty}$ and $h_{\infty}$, respectively, then by virtue of (\ref{linear}) the sequence $f_{n}$ converges to some vector $f_{\infty}$. Since $R(z)$ is a bounded operator, then $g_{\infty}=R(z)f_{\infty}$. Therefore and $g_{\infty}$ belongs to the domain of $A_{z}$ and $$A_{z}g_{\infty}=f_{\infty}+zg_{\infty}=\underset{n\rightarrow\infty}{\lim}\left(f_{n}+zg_{n}\right)=\underset{n\rightarrow\infty}{\lim}A_{z}g_{n}=h_{\infty},$$ that is $A_{z}$ is a closed operator.

According to (\ref{hilb1}) for any $z_1,z_2\in\mathbb{D}$  we have  $R(z_1)R(z_2)=R(z_2)R(z_1)$ and $$R(z_2)\mathcal{H}=R(z_1)[I+(z_2 - z_1)R(z_2)]\mathcal{H}\subseteq R(z_1)\mathcal{H}.$$ Hence the domains of all $A_z$ coincide.  

Taking some $g=R(z_1)f=R(z_2)[I+(z_{1}-z_{2})R(z_{1})]f, \; f\in\mathcal{H},$ we obtain with account of (\ref{linear}), (\ref{hilb1}) that 
\begin{equation*}\begin{array}{c}
A_{z_2}g=  \left[I+(z_{1}-z_{2})R(z_{1})\right]f+z_2 R(z_2)\left[I+(z_{1}-z_{2})R(z_{1})\right]\\
=A_{z_1}g+z_2\left\{-R(z_1)f +R(z_2)f+ (z_{1}-z_{2})R(z_{1})R(z_2)f \right\}=A_{z_1}g.
\end{array}
\end{equation*}
 Therefore $A_z$ doesn't depend on $z$ and for the operator  $A\equiv A_z, \;z\in \mathbb{D},$ by construction   
 \begin{equation}\label{resolv1}
 (A-zI)^{-1}=\left( A_z -zI \right)^{-1} =R(z),
 \end{equation}
 which is the desired conclusion.
 
 The proof of "only if" is trivial.   
\end{proof}

\begin{thm}\label{next}
If $R(z)$ as in Theorem \ref{initial} and in addition  \begin{equation} \label{symm}
  z\in\mathbb{D} \leftrightarrow \bar{z}\in \mathbb{D},
\end{equation}
 \begin{equation}\label{selfadj}
 R(\bar{z})=R(z)^{*}, \quad z\in \mathbb{D},
\end{equation} 
 then $R(z)$ is the resolvent of a selfadjoint  operator $A$.
 \end{thm}
\begin{proof}
By (\ref{linear}) for each non-real $z\in \mathbb{D}$ the range $A_z-zI\equiv A -zI$ coincides with $\mathcal{H}$. Therefore it suffices to show that $A$ is a symmetric operator. But for any $g_{1}=R(z)f_1,\;g_{2}=R(z)f_2, \;f_1,f_2 \in \mathcal{H}, \;z\in \mathbb{D},$ by virtue of  (\ref{selfadj}) and (\ref{hilb1})
\begin{equation*}\begin{array}{c}
\left(Ag_1,g_2 \right)-\left(g_1,Ag_2 \right)=\left([f_1 +zR(z)f_1 ],R(z)f_2 \right)-\left(R(z)f_1,[f_2 +zR(z)f_2]\right)\\ 
=\left([R(\bar{z}) +zR(\bar{z})R(z)-R(z)-\bar{z}R(\bar{z})R(z)]f_1,f_2 \right)=0.
\end{array}
\end{equation*}
\end{proof} 

\begin{thm}\label{simplvers}
Let $A$ be a selfadjoint operator in $\mathcal{H}$;  $R(z),\; \mathrm{Im}z\neq 0,$ is the resolvent of $A$; $f_{1},...,f_{N}< \: 1\leq N \leq \infty,$ are linearly independent vectors from $\mathcal{H}$; $Q(z)$ is the Nevanlinna $N\times N$-matrix function with elements \begin{equation}\label{gramm}
q_{mn}(z)=\left(R(z)f_{n},f_{m} \right), \;1\leq m,n \leq N. 
\end{equation}
Then for any invertible Hermitian $N\times N$ matrix $W=\left(w_{mn} \right)_{1}^{N}$ the matrix $Q(z)+W, \:\mathrm{Im}z\neq 0, $ is invertible and the operator function
\begin{equation}\label{krein1}
R_{1}(z)=R(z)-\sum\limits_{m,n=1}^{N}\left(\left[Q(z)+W \right]^{-1} \right)_{mn} \left(\cdot,R(\bar{z})f_n \right)R(z)f_m 
\end{equation}
is the resolvent of some selfadjoint operator $A_{1}$. 
\end{thm}
\begin{proof}
By our assumptions $Q(z)$ as well as $Q(z)+W$ are Nevanlinna matrix functions the imaginary parts of which $$\frac{1}{2i}\left[ Q(z)-Q(z)^{*}\right] $$ has property
\begin{equation*}
\frac{1}{z-\bar{z}}\left[ Q(z)-Q(z)^{*}\right] =\left((R(z)f_{m},R(z)f_{n}) \right)_{m,n=1}^{N}  
\end{equation*}
\begin{equation*}
=\Gamma(R(z)f_{1},...,R(z)f_{N})   \geq \lambda_{\min}(R(z)f_{1},...,R(z)f_{N})I ,
\end{equation*}
where $\Gamma(R(z)f_{1},...,R(z)f_{N})$ is the Gramm-Schmidt matrix for vectors $R(z)f_{1},...R(z)f_{N}$ and $\lambda_{\min}(R(z)f_{1},...,R(z)f_{N})$ is the minimal eigenvalue of $\Gamma(R(z)f_{1},...,R(z)f_{N})$. Since vectors $f_{1},...f_{N}$ are linerly independent and $\ker{R(z)}=\{0\}$, then $\lambda_{\min}(R(z)f_{1},...,R(z)f_{N})>0.$ Hence $Q(z)+W$ is invertible.

Suppose that $R_{1}(z)h=0,\:\mathrm{Im}z\neq 0$ for some $h\in \mathcal{H}$, that is
\begin{equation}\label{aux1}
R(z)h=\sum\limits_{m,n=1}^{N}\left(\left[Q(z)+W \right]^{-1} \right)_{mn} \left(h,R(\bar{z})f_n \right)R(z)f_m.
\end{equation}Hence $R(z)h$ is a linear combination of vectors $R(z)f_{1},...,R(z)f_{N}$ and in view of invertibility of $R(z)$ we see that $h=\alpha_{1}f_{1}+...+\alpha_{N}f_{N}$ with some coefficients $\alpha_{1},..,\alpha_{N}$. By (\ref{krein1}) 
\begin{equation*}
R_{1}(z)f_{j}=\sum\limits_{m=1}^{N}\left(\left[Q(z)+W \right]^{-1}W \right)_{mj} R(z)f_{m}, \quad j=1,...,N.
\end{equation*}
Therefore
\begin{equation}\begin{array}{c}\label{coeff1}
0=R_{1}(z)h=R(z)\left( \beta_{1}f_{1}+...+\beta_{N}f_{N}\right), \\
 \beta_{m}=\sum\limits_{m=1}^{N}\left(\left[Q(z)+W \right]^{-1}W \right)_{mj}\alpha_{j}.
\end{array}
\end{equation} 
Since $\ker R(z)=\{0\}$ and $f_{1},...f_{N}$ are linearly independent, then $\beta_{1}=...=\beta_{N}=0$. But if $W$ is invertible then by by virtue of invertibility of $Q(z)+W$ and (\ref{coeff1}) $\alpha_{1}=...=\alpha_{N}=0$, that is $h=0$. Hence $\ker R_{1}(z)=0$.

$R_{1}(z)^{*}=R_{1}(\bar{z})$ follows directly from (\ref{krein1}) because $R(z)^{*}=R(\bar{z}),\:  Q(z)^{*}=Q(\bar{z}),\:W^{*}=W.$ 

Taking into account that for $R(z)$ the Hilbert identity holds and that
\begin{equation*}\begin{array}{c}
q_{mn}(z_{2})-q_{mn}(z_{1})=\left[q_{mn}(z_{2})+w_{mn}\right]-\left[q_{mn}(z_{1})+w_{mn}\right]\\ =(z_{2}-z_{1})\left(R(z_{1})f_{n},R(\bar{z}_{2})f_{m}) \right), \quad 1\leq m,n\leq N, \quad \mathrm{Im}z_{1},z_{2}\neq 0, 
\end{array}
\end{equation*}  
one can easily verify by elementary algebraic manipulations that for $R_{1}(z)$ the Hilbert identity also holds.

We see that $R_{1}(z)$ satisfies all conditions of Theorem \ref{initial} and \ref{next}, which is the desired conclusion.
\end{proof}

\begin{rem}
Comparing the formal inverse for operators in the left and right parts of (\ref{krein1}) yields
\begin{equation}\label{triv1}
A_{1}=A+\sum\limits_{m,n=1}^{N}\left(W^{-1} \right)_{mn} \left(\cdot ,f_{n} \right)f_{m}\, , 
\end{equation}
that is if $N<\infty$, then $A_{1}$ is a finite dimensional perturbation of $A$. 
\end{rem}

\begin{rem}
Let $W$ in (\ref{krein1}) isn't invertible and 
\begin{equation*}
\mathcal{A} =\left\lbrace h\in\mathcal{H}: h=\alpha_{1}f_{1} +...+\alpha_{N}f_{N}, \;(\alpha_{1},...,\alpha_{N})^{T}\in 
\ker W \right\rbrace .
\end{equation*}
Then $R_{1}(z)h\equiv 0$ for any $h\in \mathcal{A}$ but in this case the restriction of $R_{1}(z)$ on the subspace 
$\mathcal{A}^{\bot}=\mathcal{H}\ominus\mathcal{A}$ is the resolvent of selfadjoint operator $A_{1}$ in 
$\mathcal{A}^{\bot}$.

Indeed, in the course of proof of Theorem \ref{initial} it was actually shown that $\ker R_{1}(z)=\mathcal{A}$. Since 
$R_{1}(z)^{*} =R_{1}(\bar{z})$, then $R_{1}(z)\mathcal{A}^{\bot} \subseteq \mathcal{A}^{\bot}$ and for the restriction 
of $R_{1}(z)$ on the invariant subspace $\mathcal{A}^{\bot}$ all conditions of Theorems \ref{initial} and \ref{next} 
hold.

Specifically, if $W=0$ in (\ref{krein1}) and $A$ is a bounded operator, then 
$P_{\mathcal{A}^{\bot}}R_{1}(z)|_{\mathcal{A}^{\bot}}$ where $P_{\mathcal{A}^{\bot}}$ is the orthogonal projector on 
$\mathcal{A}^{\bot}$ is the resolvent of selfadjoint operator $P_{\mathcal{A}^{\bot}}A|_{\mathcal{A}^{\bot}}$ in 
$\mathcal{A}^{\bot}$.
\end{rem}  
\section{M.G. Krein's line of argument}

M.G. Krein was the first who realized that the statement of Theorem \ref{simplvers} can be strengthened in the following way.
\begin{thm}\label{kreinmain}
Let $A$ be an unbounded selfadjoint operator in $\mathcal{H}$ and  $R(z),\; \mathrm{Im}z\neq 0,$ is the resolvent of $A$;  $\{g_{n}(z)\}_{n=1}^{N}, \: 1\leq N \leq \infty,$  is the set of $\mathcal{H}$-valued holomorphic in the open upper and lower half-planes vector functions satisfying the conditions
\begin{itemize}
\item{for any non-real $z,z_{0}$
\begin{equation}\label{res3}
g_{n}(z)=g_{n}(z_{0})+(z-z_{0})R(z)g_{n}(z_{0}), \; j=1,...,N; 
\end{equation}
}
\item{at least for one non-real $z_{0}$ vectors $\{g_{n}(z_{0})\}_{n=1}^{N}$ form a basis $($ Riesz basis if $N=\infty$ $)$  in their $($closed if $N=\infty$$)$ linear span $\mathcal{N}$ and none of non-zero vectors from $\mathcal{N}$ belongs to the domain $\mathcal{D}(A)$ of $A$};
\end{itemize}
$Q(z)$ is a holomorphic in the open upper and lower half-planes $N\times N$-matrix function $($that generates a bounded operator in the space $\mathfrak{l}_{2}$ if $N=\infty$ $)$ such that 
\begin{itemize}
\item{$Q(z)^{*}=Q(\bar{z}), \; \mathrm{z}\neq 0$;}
\item{for any non-real $z,z_{0}$
\begin{equation}
Q(z)-Q(z_{0})=(z-z_{0})\left( \left(g_{m}(z),g_{n}(\bar{z_{0}})\right)\right)^{T}_{1\leq m,n\leq N}. 
\end{equation}
}
\end{itemize} 
Then for any Hermitian $N\times N$ matrix $W=\left(w_{mn} \right)_{1}^{N}$ $($  such that the closure of  linear operator defined as multiplication by $W$ on a set of $\mathfrak{l}_{2}$-vectors with a finite number of non-zero coordinates is a selfadjoint operator in $\mathfrak{l}_{2}$ if $N=\infty$ $)$ the matrix (operator in $\mathfrak{l}_{2}$ if $N=\infty$ $)$ $Q(z)+W, \:\mathrm{Im}z\neq 0, $ is $($boundedly$)$ invertible and the operator function
\begin{equation}\label{krein4}
R_{1}(z)=R(z)-\sum\limits_{m,n=1}^{N}\left(\left[Q(z)+W \right]^{-1} \right)_{mn} \left(\cdot,g_{n}(\bar{z}) \right)g_{m}(z)
\end{equation}
is the resolvent of some selfadjoint operator $A_{1}$. 
\end{thm}  

\begin{proof}
If $\{g_{n}(z_{0})\}_{n=1}^{N}$ is a (Riesz) basis in  $\mathcal{N}$ for some non-real $z_{0}$, then $\{g_{n}(z)\}_{n=1}^{N}$ is a (Riesz) basis in  $\mathcal{N}$ for any non-real $z$. Indeed, by (\ref{res3}) 
\begin{equation}\label{res4}
g_{n}(z)=U_{z_{0}}(z)g_{n}(z_{0}), \quad U_{z_{0}}(z)=\left(A-z_{0}I \right)\cdot\left(A-zI \right)^{-1},
\end{equation}
and for any non-real $z,z_{0}$ the operator $ U_{z_{0}}(z)$ is bounded and boundedly invertible.

The invertibility of $Q(z)+W$ can be expressis verbis proved as in Theorem \ref{simplvers} if one remembers that in the limit case $N=\infty$ for any Riesz basis, in particular for $\{g_{n}(z)\}_{n=1}^{\infty}, \: \mathrm{Im}\neq 0, $ the  corresponding infinite Gramm-Schmidt matrix generates a bounded, positive and boundedly invertible operator in $\mathfrak{l}_{2}$ ( see, for example, \cite{GoKr}).

Suppose that there is a vector $h\in \mathcal{H}$  such that $R_{1}(z_{0})h=0$ for some non-real $z_{0}$. By (\ref{krein4}) this means that
\begin{equation}\label{supp1}\begin{array}{c}
R\left(z_{0}\right)h\left(\in\mathcal{D}(A)\right)=\sum\limits_{m,n=1}^{N}\left(\left[Q(z)+W \right]^{-1} \right)_{mn} \left(h,g_{n}(\bar{z_{0}}) \right)g_{m}(z_{0})\\
=\sum\limits_{m=1}^{N}\left\{\sum\limits_{n=1}^{N}\left(\left[Q(z)+W \right]^{-1} \right)_{mn} \left(h,g_{n}(\bar{z_{0}}) \right)\right\}g_{m}(z_{0}) \in\mathcal{N}.  
\end{array}
\end{equation}
But for any $h\in\mathcal{H}$ the vector in the left hand side of (\ref{supp1}) belongs to $\mathcal{D}(A)$ while the corresponding vector in the right hand side of (\ref{supp1}) belongs to $\mathcal{N}$. However, by our assumptions $\mathcal{D}(A)\cap \mathcal{N}=\{0\}$. Hence both sides of (\ref{supp1}) are zero-vectors, particularly $R\left(z_{0}\right)h=0$. Recalling that  the resolvent  $R\left(z_{0}\right)$ of selfadjoint operator $A$ is invertible, we conclude that $h=0$. 

The property $R_{1}(z)^{*}=R_{1}(\bar{z}),\, \mathrm{Im}z\neq 0$ is evident.

The fact that $R_{1}(z)$ satisfies the Hilbert identity for any two non-real $z_{1},z_{2}$ can be checked out by elementary algebraic computation.

\end{proof}
The following theorem extends the class of singular perturbations of selfadjoint operators.
\begin{thm}\label{kreinnext}
Let $\mathcal{H}$ and $\mathcal{K}$ be Hilbert spaces, $A$ be an unbounded selfadjoint operator in $\mathcal{H}$ and  $R(z),\; \mathrm{Im}z\neq 0,$ is the resolvent of $A$, $G(z)$  is a bounded holomorphic in the open upper and lower half-planes operator function from $\mathcal{K}$ to  $\mathcal{H}$ satisfying the conditions 
\begin{itemize}
\item{for any non-real $z,z_{0}$
\begin{equation}\label{res3}
G(z)=G(z_{0})+(z-z_{0})R(z)G(z_{0}), 
\end{equation}
}
\item{at least for one and hence for all non-real $z$ zero is not an eigenvalue of the operator  $G(z)^{*}G(z)$ and the intersection of the domain $\mathcal{D}(A)$ of $A$ and the subspace $\mathcal{N}=\overline{G(z_{0})\mathcal{K}}\subset\mathcal{H}$ consists only of the zero-vector;}
\end{itemize}
$Q(z)$ is a holomorphic in the open upper and lower half-planes operator function in $\mathcal{K}$ such that 
\begin{itemize}
\item{$Q(z)^{*}=Q(\bar{z}), \; \mathrm{z}\neq 0$;}
\item{for any non-real $z,z_{0}$
\begin{equation}
Q(z)-Q(z_{0})=(z-z_{0})G(\bar{z_{0}})^{*}G(z). 
\end{equation}
}
\end{itemize} 
Then for any invertible selfadjoint operator $L$ in $\mathcal{K}$ such that $L^{-1}$ is compact the operator $Q(z)+L, \:\mathrm{Im}z\neq 0, $ is invertible, has compact inverse and the operator function
\begin{equation}\label{krein5}
R_{L}(z)=R(z)-G(z)\left[Q(z)+L \right]^{-1} G(\bar{z})^{*}
\end{equation}
is the resolvent of some selfadjoint operator $A_{1}$. 
\end{thm}
\begin{proof}\label{wider}
Suppose that for some non-real $z_{0}$ zero is not an eigenvalue of $G(z_{0})^{*}G(z_{0})$ and at the same time  there are a non-real $z_{1}$ and and a non-zero $h\in\mathcal{K}$ such that $G(z_{1})h=0$. Then by (\ref{res3}) $$G(z_{0})h=\left[I+(z_{0}-z_{1})R(z_{0})\right]R(z_{1})h=0,$$
a contradiction.

By our assumptions for any non-real $z$ zero is not an eigenvalue of operator $Q(z)+L$. Indeed, suppose that for some $h$ from the domain of $L$ we have $\left[Q(z)+L \right]h=0$. Then $$0=\mathrm{Im}\left(\left[Q(z)+L \right]h,h \right)=\mathrm{Im}\cdot\left(Q(z)^{*} Q(z)h,h\right).$$
But $Q(z)^{*} Q(z)$ is a non-negative invertible operator. Hence $h=0$. 

Since for non-real $z$ zero is not an eigenvalue of the operator $Q(z)+L$, by virtue of the invertibility of the operator $L$, compactnes of $L^{-1}$ and the obvious equality
$$Q(z)+L=L\left[L^{-1}Q(z)+I \right] $$ "-1" is not an eigenvalue of operator $L^{-1}Q(z)$. But $L^{-1}Q(z)$ is a compact operator. Therefore the operator  $L^{-1}Q(z)+I$ is boundedly invertible \cite{Kat} and so is the operator $Q(z)+L$, $$\left[ Q(z)+L\right]^{-1}=\left[L^{-1}Q(z)+I \right]^{-1} L^{-1} $$  Evidently, the inverse of $Q(z)+L$ is a compact operator.

The fact that $R_{1}(z)$ is the resolvent of a self-adjoint operator is proved by the same arguments as above. 
\end{proof}

\section{Singular perturbations of selfadjoint Laplace operator. Null-range potentials}
Let $A$ be an unbounded selfadjoint operator. By a \textit{regular} perturbation of $A$ we call any selfadjoint operator $A_{1}$ defined as in Theorem \ref{simplvers} . Following \cite{AK} we say that $A_{1}$ is a \textit{singular} perturbation of $A$ if $A_{1}$ is defined by $A$ as in Theorem \ref{kreinmain} or in Theorem \ref{kreinnext}.
In this Section we will consider a special class of singular perturbations of the selfadjoint Laplace operator  $$-\Delta=-\frac{\partial^{2}}{\partial x^{2}_{1}}-\frac{\partial^{2}}{\partial x^{2}_{2}}-\frac{\partial^{2}}{\partial x^{2}_{3}} $$  in $\mathbf{L}_{2}(\mathbf{R}_{3})$  defined on the Sobolev subspaces ${H}_{2}^{2}\left( \mathbf{R}_{3}\right)$, namely, the class of operators which fit into the conditions of Theorem \ref{kreinmain}. We will use here the symbol $A$ to denote the specified unperturbed Laplace operator and the symbol $R(z)$ to denote the resolvent of $A$. Remind that  
\begin{equation}\label{lapl}
 \left(R(z)f\right)(\mathbf{x})=\frac{1}{4\pi}
  \int_{\mathbf{R}_{3}}\frac{e^{i\sqrt{z}|\mathbf{x}-\mathbf{x}^\prime|}}{|\mathbf{x}-\mathbf{x}^\prime|}f\left(\mathbf{x}^\prime \right)d\mathbf{x}^\prime , \quad \mathrm{Im}{\sqrt{z}}>0, \quad \mathbf{x}=\left(x_{1},x_{2},x_{3} \right) .
 \end{equation}
 A simple but fundamentally important example of singular perturbation of $A$ was first rigorously examined in the short note \cite{BF} in the framework of the theory of self-adjoint extensions of symmetric operators. Actually, it was proved in \cite{BF} that  for   $$g(z;\mathbf{x})=\left(R(z)\delta\right)(\mathbf{x})=\frac{1}{4\pi}\frac{e^{i\sqrt{z}|\mathbf{x}|}}{|\mathbf{x}|},$$ where  $\delta\left( \mathbf{x}\right)$ is the Dirac $\delta$-function, 
 and  for any real $\alpha$ the operator function 
\begin{equation}\label{examp1}
R_{\alpha}(z)=R(z)-\frac{1}{Q(z)+\alpha}\left(\cdot,g(\bar{z};\cdot) \right)g(z;\cdot), \quad Q(z)=\frac{i\sqrt{z}}{4\pi},
\end{equation}  
is the resolvent of selfadjoint operator $A_{\alpha}$. In accordance with (\ref{examp1})   the domain $\mathcal{D}_{\alpha}$ of $A_{\alpha}$ consists of functions 
\begin{equation}\label{alpha1} 
f(\mathbf{x})=f_{0}(\mathbf{x})-\frac{1}{Q(z)+\alpha}\cdot f_{0}(\mathbf{0})\cdot g(z;\mathbf{x}),
\end{equation} 
where functions $f_{0}(\mathbf{x})$ run the space ${H}_{2}^{2}$ and 
\begin{equation}\label{alpha2} 
\left( A_{\alpha}f\right) (\mathbf{x})=\left( Af_{0}\right)(\mathbf{x})-\frac{z}{Q(z)+\alpha}\cdot f_{0}(\mathbf{0})\cdot g(z;\mathbf{x}).
\end{equation}
The expressions (\ref{alpha1}),(\ref{alpha2}) are correct since any vector $\hat{f}_{0}$ of ${H}_{2}^{2}\left( \mathbf{R}_{3}\right)$ is equivalent to some H\"{o}lder continuous function $f_{0}(\mathbf{x})$ with any index $\gamma <\dfrac{1}{2}$ \cite{Kat} and consequently the product $|\mathbf{x}|\cdot f_{0}(\mathbf{x})$ is differentiable in $|\mathbf{x}|$ at $\mathbf{x}=\mathbf{0}$ and 
\begin{equation}\label{different}
\underset {|\mathbf{x}|\downarrow 0}{\lim} \frac {\partial} {\partial|\mathbf{x}|}\left(|\mathbf{x}|f(\mathbf{x})\right)=f_{0}(\mathbf{0}). 
\end{equation}
With account of (\ref{alpha1}) and (\ref{different}) it can be argued that  functions $f(\mathbf{x})$ from
$\mathcal{D}_{\alpha}$ satisfy the "boundary condition" 
\begin{equation}\label{simplest}
\underset {|\mathbf{x}|\downarrow 0}{\lim} \left[\frac {\partial} {\partial|\mathbf{x}|}\left(|\mathbf{x}|f(\mathbf{x})\right)+\alpha |\mathbf{x}|f(\mathbf{x}) \right]=0. 
\end{equation} 

For real $\alpha$ the selfadjoint operator $A_{\alpha}$  legalizes the formal expression $-\Delta+\alpha\cdot \delta(\mathbf{x}) $ and 
associated with $A_{\alpha}$ the condition (\ref{simplest}) is said to be a null-range potential \cite{BF}, \cite{AGHH}.
 
 Note that the $g(z;\cdot)$ in (\ref{examp1}) doesn't belong to $\mathcal{D}(A)$, otherwise the functional
 \begin{equation}\label{smfin}
 \varphi(0)=
  - \int_{\mathbf{R}_{3}}\left[ \left( \Delta\varphi\right) \left(\mathbf{x} \right)+z\varphi\left(\mathbf{x} \right)\right]\cdot\bar{g(z;\mathbf{x})} d\mathbf{x} 
 \end{equation}
 would be bounded in the set of infinitely smooth compact function $\varphi\left(\mathbf{x} \right)$. Besides,
 \[
  (z-z_{0}) \left(g(z),g(\bar{z_{0}})\right)=\underset{|\mathbf{x}|\downarrow 0}{\lim} \left[g(z;\mathbf{x})-g(z_{0};\mathbf{x})\right] =Q(z)-Q(z_{0}).
  \]
 Therefore the adduced result from \cite{BF} is a special case of Theorem \ref{kreinmain}, where $A$ is the standardly defined Laplace operator and $N=1$. 

Referring to the conditions of Theorem (\ref{kreinmain}), it is easy to check that the stated assertion about $R_{\alpha}(z)$ admits the following (in fact, well-known \cite{AGHH}, \cite{AK}) generalization.
 Let 
 \begin{equation}\label{genW}\begin{array}{c}
    g_{n}(z;\mathbf{x})= R(z)\delta(\cdot -\mathbf{x}_{n})(\mathbf{x})=\frac{1}{4\pi}\frac{e^{i\sqrt{z}|\mathbf{x}-\mathbf{x}_{n}|}}{|\mathbf{x}-\mathbf{x}_{n}|}, \quad 1<n\leq N<\infty;\\
  Q(z)=\left(q_{mn}(z) \right)_{m,n=1}^{N}=\left\lbrace \begin{array}{c} 
  q_{mn}(z)=g_{n}(z;\mathbf{x}_{m}-\mathbf{x}_{n}), \quad m\neq n,\\
  q_{mm}(z)=\frac{i\sqrt{z}}{4\pi}
  \end{array}\right. . 
  \end{array} \end{equation}
 Using the same arguments as above, it is easy to check that any non-zero linear combination of functions $g_{n}(z;\mathbf{x})$ doesn't belong to $\mathcal{D}(A)$. Besides,  $Q(z)$ is a holomorphic in the open upper and lower half-planes infinite matrix function defining at each non-real $z$ a bounded operator in the space $\mathbf{l}_{2}$ such that 
  \begin{itemize}
  \item{$Q(z)^{*}=Q(\bar{z}), \; \mathrm{z}\neq 0$;}
  \item{for any non-real $z,z_{0}$
  \begin{equation}
  Q(z)-Q(z_{0})=(z-z_{0})\left( \left(g_{m}(z),g_{n}(\bar{z_{0}})\right)\right)^{T}_{1\leq m,n<\infty}. 
  \end{equation}
  }
  \end{itemize} 
 As follows, dy virtue of Theorem \ref{kreinmain}  for any invertible Hermitian matrix $W=\left(w_{mn} \right)_{m,n=1}^{N}$ the operator function 
  \begin{equation}\label{exampW}
     R_{\alpha}(z)=R(z)-\sum\limits_{m,n=1}^{N}\left([ Q(z)+W]^{-1}\right)_{mn} \left(\cdot,g_{n}(\bar{z};\cdot) \right)g_{m}(z;\cdot)\end{equation} is the resolvent of selfadjoint operator $-A_{W}$ in $\mathbf{L}_{2}(\mathbf{R}_{3})$.
  
 Let us denote by $\mathcal{N}$ the linear span of functions $\{g_{n}(z;\mathbf{x})\}$.  The operator $A_{W}$ is loosely speaking the Laplace differential operator $-\Delta$ with the domain 
 
\begin{equation}\label{simpleW}
\begin{array}{c}      
     \mathcal{D}_{W}:=  \left\lbrace f: \, f=f_{0}+g, \, f\in {H}_{2}^{2}\left(\mathbf{R}_{3}\right), g\in\mathcal{N}, \right. \\  
        \underset{\rho_{m}\rightarrow 0}{\lim}\left[\frac {\partial} {\partial\rho_{m}}\left(\rho_{m}f(\mathbf{x})\right)\right]+\sum\limits_{n=1}^{N}w_{mn} 
\underset {\rho_{n}\rightarrow 0}{\lim}\,[\rho_{n}\,f(\mathbf{x})] =0, \\ \left. \rho_{n}=|\mathbf{x}-\mathbf{x}_{n}|, \quad 1\leq n\leq N  \right\rbrace .  \end{array} 
 \end{equation} 
  
  If the matrix $W$ is diagonal, that is $w_{mn}=\alpha_{m}\cdot \delta_{mn}$, then $A_{W}$ is the Laplace operator perturbed by a collection of "null-range" potentials
  $$\underset{\rho_{m}\rightarrow 0}{\lim}\left[ \frac {\partial} {\partial\rho_{m}}\left(\rho_{m}f(\mathbf{x})\right)+ \alpha_{m}\cdot\rho_{m}f(\mathbf{x})\right] =0. $$
  
 With some reservations the last statements remain true also in the case of the infinite set of points $\{\mathbf{x}_{n}\}_{-\infty}^{\infty}$. Let the set of functions $ g_{n}(z;\mathbf{x})$ and infinite matrix function $Q(z)$ be like in $($\ref{genW}$)$ and $\mathcal{N}$ denotes the closed linear span of functions $ g_{n}(z;\mathbf{x})$.
 \begin{thm}[A. Grossmann, R. H\o{}egh-Krohn, M. Mebkhout]\label{kreinmain+}
 
If $$ \underset{-\infty<m,n<\infty}{\inf}\left|\mathbf{x}_{m}-\mathbf{x}_{n} \right|=d>0,   $$  then for a selfadjoint operator  in $\mathbf{l}_{2}$  defined by the infinite matrix $W=\left( w_{mn}\right)_{m,n=-\infty}^{\infty} $ the operator function
  $$ R_{W}(z)=R(z)-\sum\limits_{m,n=-\infty}^{\infty}\left([ Q(z)+W]^{-1}\right)_{mn} \left(\cdot,g_{n}(\bar{z};\cdot) \right)g_{m}(z;\cdot) $$ is the resolvent of selfadjoint operator $-\Delta_{W}$ in $\mathbf{L}_{2}(\mathbf{R}_{3})$, which is the Laplace operator with the domain
   \[
      \mathcal{D}_{W}:=  \left\lbrace f: \, f=f_{0}+g, \, f_{0}\in {H}_{2}^{2}\left(\mathbf{R}_{3}\right), \, g\in \mathcal{N}, \right.  \] 
      
      \[ \left. \begin{array}{c} \underset{\rho_{m}\rightarrow 0}{\lim}\left[\frac {\partial} {\partial\rho_{m}}\left(\rho_{m}f(\mathbf{x})\right)\right]+\sum\limits_{n=-\infty}^{\infty}w_{mn} \underset
   {\rho_{n}\rightarrow 0}{\lim}\,[\rho_{n}\,f(\mathbf{x})] =0, \\ \rho_{n}=|\mathbf{x}-\mathbf{x}_{n}|, \quad -\infty\leq n< \infty .  \end{array}\right.  \]
  \end{thm}
Theorem \ref{kreinmain+} is a direct consequence of  Theorem \ref{kreinmain} and the following proposition.
\begin{prop}
If \begin{equation}\label{mindist}
\underset{m,n}{\inf}\left|\mathbf{x}_{m}-\mathbf{x}_{n} \right|=d>0   
\end{equation}  
and $\mathrm{Im} z\neq 0,$ then $Q(z)$ is the matrix of bounded operator in the natural basis of $\mathbf{l}_{2}$ and the sequence of $\mathbf{L}_{2}(\mathbf{R}_{3})$-vectors $\left\lbrace g_{n}(z;\cdot) \right\rbrace_{1}^{\infty}$ is the Riesz basis in its closed linear span $\mathcal{N}$.
\end{prop}
\begin{proof}
By (\ref{genW}) and (\ref{mindist}) for $\mathrm{Im}\sqrt{z}=\eta+i\kappa$ with $\kappa>0$ we see that 
\begin{equation*}
\sum\limits_{n}^{}\left|q_{mn}(z) \right|\leq \frac{\sqrt {\eta^{2}+\kappa^{2} }}{4\pi}+\frac{1}{4\pi d}\sum\limits_{n\neq m}^{}e^{-\kappa\left|\mathbf{x}_{n}-\mathbf{x}_{m} \right|}<\infty .
\end{equation*}
and noting that there are at most $3n^{2}+\frac{1}{4}$ points $\mathbf{x}_{n}$ in the spherical layer $(n-\frac{1}{2})d\leq\left|\mathbf{x}-\mathbf{x}_{m} \right|<(n+\frac{1}{2}), \, n\geq 1,$ 
obtain that 
\begin{equation}\label{estim1}\begin{array}{c}
\sum\limits_{n\neq m}^{}e^{-\kappa\left|\mathbf{x}_{n}-\mathbf{x}_{m} \right|}\leq \frac{13}{4}e^{-\kappa d}+\sum\limits_{2}^{\infty}\left(3n^{2}+\frac{1}{4} \right)e^{-\left( n-\frac{1}{2}\right)\kappa d } \\ \underset{\kappa\rightarrow\infty}{=}\leq \frac{13}{4}e^{-\kappa d}+O\left(e^{-\frac{3}{2}\kappa d} \right) ). \end{array} 
\end{equation}
Hence for non-real $z$ 
\begin{equation*}
M(z)=\underset{m}{\sup}\sum\limits_{n}^{}\left|q_{mn}(z) \right|<\infty
\end{equation*}
and the infinite matrix $Q(z)$ generates a bounded operator in $\mathbf{l}_{2}$ with norm $\left\|Q(z)\right\| \leq M(z)$.

As was mentioned in the proof of  Theorem \ref{kreinmain}, in order to establish that for any regular point $z$ of the Laplace operator $A$ the set of $\mathbf{L}_{2}(\mathbf{R}_{3})$-functions $ \{g_{n}(z)=g(z;\mathbf{x}-\mathbf{x}_{n} )\}$ forms a Riesz basis in its linear span, it suffices to verify this for at least one such point, say for a point $-\kappa ^{2}$, where $\kappa$ is a sufficiently large positive number. For $z=-\kappa^{2}$ the Gramm-Schmidt matrix for the set of functions $ \{g_{n}(-\kappa^{2})$ has form 
\begin{equation}\label{gramm1}
\Gamma(-\kappa^{2})=\frac{1}{8\pi\kappa}\left(e^{-\kappa\left|\mathbf{x}_{m}-\mathbf{x}_{n} \right| } \right)_{-\infty}^{\infty} .
\end{equation}
By (\ref{gramm1}) the matrix $8\pi\kappa\cdot\Gamma(-\kappa^{2})$ is the sum $I+\Delta(-\kappa^{2})$ of the infinite unity matrix $I$ and the matrix $\Delta(-\kappa^{2})$, which  according to (\ref{estim1}) generates a bounded operator in $\mathbf{l}_{2}$ with norm of less than one for for sufficiently large $\kappa$. Therefore the  matrix $\Gamma(-\kappa^{2})$ generates a bounded and boundedly invertible operator in $\mathbf{l}_{2}$. Hence vectors $ \{g_{n}(-\kappa^{2})$ form a Riesz basis in their linear span.  
\end{proof}
\section{Singular perturbations of selfadjoint Laplace operator. 1D-located perturbation}
We describe further a special class of singular selfadjoint  perturbations of the Laplace operator $A$ falling under the conditions of Theorem \ref{kreinnext}. In the cases discussed below, $\mathbf{L}_{2}(\mathbf{R}_{3})$ plays naturally the role of Hilbert space $\mathcal{H}$, the usual space $\mathbf{L}_{2}([0,l])$ of square integrable functions on the interval $[0,l]$ with $l<\infty$ appears as the Hilbert space $\mathcal{K}$   wherein this interval itself is identified with the subset $\mathbf{\mathfrak{l}}=\left\lbrace 0\leq x_{1}\leq l, \, x_{2}=0, \,  x_{3}=0 \right\rbrace $  of $\mathbf{R}_{3}$.
 We define the holomorphic operator function $G(z), \, \mathrm{Im}(z)\neq 0$ from $\mathbf{L}_{2}([0,l])$ to $\mathbf{L}_{2}(\mathbf{R}_{3})$ setting
 \begin{equation}\label{stl}\begin{array}{c}
 \left(G(z)u\right)(\mathbf{x})=\int\limits_{0}^{l}g\left(z|x_{1},x_{2},x_{3};x^\prime_{1} ,0,0 \right)u(x^\prime_{1})dx^\prime_{1}, \quad u(\cdot)\in \mathbf{L}_{2}([0,l]),  \\   g\left(z|x_{1},x_{2},x_{3};x^\prime_{1},x^\prime_{2},x^\prime_{3} \right)=g(z|\mathbf{x},\mathbf{x}^{\prime}) =\frac{1}{4\pi}
 \frac{e^{i\sqrt{z}|\mathbf{x}-\mathbf{x}^\prime|}}{|\mathbf{x}-\mathbf{x}^\prime|}, \quad \mathrm{Im}{\sqrt{z}}>0 .
 \end{array}
 \end{equation} 
 It follows from (\ref{stl}) that
 \begin{eqnarray*}
  \left| \left(G(z)u\right)(\mathbf{x})\right|^{2} \leq \int\limits_{0}^{l}\left| g\left(z|x_{1}-x^{\prime}_{1},x_{2},x_{3};0 ,0,0 \right)\right|^{2}dx^{\prime}_{1} \cdot\left\|u\right\|^{2}  .
 \end{eqnarray*}
Therefore for $z\neq 0$ the operator $G(z)$ is bounded and
\begin{equation}\label{norm}
\left\| G(z)\right\|\leq \frac{1}{\sqrt{8\pi\mathrm{Im}\sqrt{z}}}.
\end{equation}

Note that the Fourier transform $$\widehat{G(z)u}\left(k_{1},k_{2},k_{3}\right)=\frac{1}{2\pi}\cdot\frac{1}{k_{1}^{2}+k_{2}^{2}+k_{3}^{2}-z}\hat{u}(k_{1}) $$ of $G(z)u(\mathbf{x})$, where $$\hat{u}(k_{1})=\frac{1}{\sqrt{2\pi}}\int\limits_{0}^{l}e^{-ik_{1}x_{1}}u(x_{1})dx_{1},$$  equals to zero if and only if $u(k_{1})\equiv0$ and as follows $\hat{u}(x_{1})=0$ almost everywhere on $[0,l]$. Accordingly, \textit{ $G(z)u(\cdot)=0$ in $\mathbf{L}_{2}(\mathbf{R}_{3})$ if and only if $u(\cdot)=0$ in $\mathbf{L}_{2}([0,l])$}. Therefore for any non-real ${z}$ is not an eigenvalue of $G(z)^{*}G(Z)$.
We note that the adjoint operator $G(z)^{*}$ from $\mathbf{L}_{2}(\mathbf{R}_{3})$ to $\mathbf{L}_{2}([0,l])$ is determined by an expression 
\begin{equation}\label{conj}
\left(G(z)^{*}f\right)(x)= \int_{\mathbf{R}_{3}}g(\bar{z}|x,0,0;\mathbf{x}^{\prime})f(\mathbf{x}^\prime)d\mathbf{x}^\prime, \quad f(\cdot)\in \mathbf{L}_{2}(\mathbf{R}_{3}), x\in[0,\mathfrak{l}],
\end{equation}
that makes sense, since the functions $$f(\mathbf{x})=\int_{\mathbf{R}_{3}}g(\bar{z}|\mathbf{x};\mathbf{x}^{\prime})f(\mathbf{x}^\prime)d\mathbf{x}^\prime, \quad f(\cdot)\in \mathbf{L}_{2}(\mathbf{R}_{3}), \quad \mathrm{Im}z\neq 0, $$
forming the domain $\mathcal{D}(A)$ of $A$ are continuous \cite{Kat}.

Suppose further that there is a vector $h\in \mathbf{L}_{2}(\mathbf{R}_{3})$ from the subspace $\mathcal{N} =\overline{G(z)\mathbf{L}_{2}([0,l])}$ that belongs to the domain $\mathfrak{D}(A)$ of the Laplace operator $A$. $h$ as any vector from  can be represented in the form $h=R(z)w$ with some $w\in\mathbf{L}_{2}(\mathbf{R}_{3})$ while by our assumption there is a sequence of vectors $\{u_{n}\in \mathbf{L}_{2}([0,l]) \}$ such that
\begin{equation}\label{lim-a}
\underset{n\rightarrow \infty}{\lim}\|R(z)w-G(z)u_{n}\|_{\mathbf{L}_{2}(\mathbf{R}_{3})}=0.
\end{equation}   
Now recall that for each $w\in\mathbf{L}_{2}(\mathbf{R}_{3})$ and any infinitesimal $\varepsilon>0$ it is possible to find an infinitely smooth compact function $\phi(\mathbf{r})$ which is also equal to zero at some neighborhood of the subset $\mathfrak{l}$ to satisfy the condition 
\begin{equation}\label{lim-b}
\left| \left(w,\phi \right)_{\mathbf{L}_{2}(\mathbf{R}_{3})} \right|\geq\left(1-\varepsilon \right)\left\| w \right\|^{2}_{\mathbf{L}_{2}(\mathbf{R}_{3}}.
\end{equation}
Taking into account further that for $\phi(\mathbf{r})$, as well as for any smooth compact function,
\begin{equation}\label{laplres}
\phi(\mathbf{r})= \frac{1}{4\pi}
\int_{\mathbf{R}_{3}}\frac{e^{i\sqrt{z}|\mathbf{x}-\mathbf{x}^\prime|}}{|\mathbf{x}-\mathbf{x}^\prime|}\left[ -\Delta\phi\left(\mathbf{x}^\prime \right)-z\phi\left(\mathbf{x}^\prime \right)\right]  d\mathbf{x}^\prime, 
\end{equation}
we notice that 
\begin{eqnarray*}\begin{array}{c}
\left(G(z)w,\left[-\Delta\phi -\bar{z}\phi \right] \right)_{\mathbf{L}_{2}(\mathbf{R}_{3})}=\left(w,\phi\right)_{\mathbf{L}_{2}(\mathbf{R}_{3})}, \\  \left(G(z)u,\left[-\Delta\phi -\bar{z}\phi \right]\right)_{\mathbf{L}_{2}(\mathbf{R}_{3})} =0, \quad u\in \mathbf{L}_{2}([0,l]).		
\end{array}
\end{eqnarray*}
Hence for the above sequence $\{u_{n}\in \mathbf{L}_{2}([0,l]) \}$ by virtue of (\ref{lim-b}) we conclude that 
\begin{equation}\label{lim-c}
\begin{array}{c}
\|R(z)w-G(z)u_{n}\|_{\mathbf{L}_{2}(\mathbf{R}_{3})}\cdot
\left\| -\Delta\phi-\bar{z}\phi  \right\|_{\mathbf{L}_{2}(\mathbf{R}_{3})} \\ \geq 
\left|\left( \left[R(z)w-G(z)u_{n}\right],\left[-\Delta\phi-\bar{z}\phi \right] 
\right)_{\mathbf{L}_{2}(\mathbf{R}_{3})}\right|  =\left|\left(w,\phi\right)_{\mathbf{L}_{2}(\mathbf{R}_{3})}\right| \\ \geq\left(1-\varepsilon\right)
\left\| w\right\|^{2}_{\mathbf{L}_{2}(\mathbf{R}_{3})}.
\end{array}
\end{equation}
But in view of (\ref{lim-a}) for $n \rightarrow\infty $ the last  inequality in (\ref{lim-c}) must necessarily be violated unless  $ w=0 $. Therefore $\mathcal{N} \cap \mathfrak{D}(A)=\{0\} $. 

In accordance with our choice (\ref{stl}) of the mapping $G(z)$ , the bounded  holomorphic operator function $Q(z)$ in $\mathbf{L}_{2}([0,l])$ in the corresponding Theorem \ref{kreinnext} may be determined by setting
\begin{equation}\label{qfunk}\begin{array}{c}
\left( Q(z)u\right)(x)=\int\limits_{0}^{l} q(z|x,x^{\prime})u(x^{\prime})dx^{\prime} \\ \equiv\frac{1}{4\pi}\int\limits_{0}^{l} \frac{e^{i\sqrt{z}\left|x-x^{\prime}\right|}-1 }{\left|x-x^{\prime}\right|}u(x^{\prime})dx^{\prime}, \quad u\in\mathbf{L}_{2}([0,l]), \quad \mathrm{Im}\sqrt{z}>0.
\end{array}
\end{equation}
Since the kernel $q(z|x,x^{\prime})$ of integral operator $Q(z)$ is a continuous function on the set $[0,l]\times[0,l]$, then for any non-positive $z$ the operator $Q(z)$ is bounded and moreover compact. 

For the operator function $Q(z)$ defined by the expression (\ref{qfunk}) the property $Q(z)^{*}=Q(\bar{z})$ is obvious and the relation (\ref{res3}) follows immediately from the Hilbert identity for the resolvent kernel of the Laplace operator $A$: 
\begin{equation*}\begin{array}{c}
g(z|\mathbf{x},\mathbf{x}^{\prime})-g(z_{0}|\mathbf{x},\mathbf{x}^{\prime})=(z-z_{0}) \int_{\mathbf{R}_{3}}g(z_{0}|\mathbf{x},\mathbf{x}^{\prime\prime})g(z|\mathbf{x}^{\prime\prime},\mathbf{x}^{\prime})d\mathbf{x}^{\prime\prime},\\ \mathrm{Im}z_{0},\mathrm{Im}z \neq 0. 
\end{array}
\end{equation*}
in cases where $\mathbf{x}=(x_{1}=x,x_{2}=0,x_{3}=0), \,\mathbf{x}^{\prime}=(x_{1}=x^{\prime},x_{2}=0,x_{3}=0)$.

Finally,  in the case under consideration we can take as $L$ in (\ref{krein5}) the selfadjoint Sturm-Liouville operator  $$L=-\frac{d^{2}}{dx^{2}}+\mathfrak{v}(x) $$
 in $\mathbf{L}_{2}([0,l])$ with a real continuous "potential" $\mathfrak{v}(x)$ assuming that the domain $\mathcal{D}(L)$ of $L$ consists of functions $u(x)$ from the Sobolev class $H^{2}_{2}([0,l])$ satisfying the boundary conditions $u(0)=u(l)=0$.  We confine ourselves also to only those potentials $\mathfrak{v}(x)$ for which zero is not an eigenvalue of the operator $L$. Since the concerned Sturm-Liouville operators are semi-bounded from below, have simple discrete spectrum and for their eigenvalues $\lambda_{n} $ numbered in increasing order, we have the relation $$\lambda_{n}\underset{n\rightarrow\infty}{=}\frac{\pi^{2}n^{2}}
 	{l^{2}}\left[1+ O(\frac{1}{n})\right], $$ then $L^{-1}$ is a compact operator of trace class.

Thus, the related to the Laplace operator operator $A$ operator functions $G(z)$ and $Q(z)$ and defined by formulas (\ref{stl}) and (\ref{qfunk}),respectively and the introduced selfadjoint Sturm-Liouville operator $L$ in $\mathbf{L}_{2}([0,l])$ satisfy all the conditions of Theorem \ref{kreinnext}. Hence, the operator function $R_{L}(z)$ defined by the expression (\ref{krein5}) is the resolvent of some singular perturbation $A_{L}$ of $A$ in $\mathbf{L}_{2}(\mathbf{R}_{3})$.
\begin{prop}\label{local1}
	Any smooth compact function $\phi(\mathbf{r})$, which is equal to zero at some neighborhood of the subset $\mathfrak{l}$ belongs to $\mathcal{D}(A_{L})$ and $$\left(A_{L}\phi\right)(\mathbf{r})=\left(A\phi\right)(\mathbf{r}) 
	=-\Delta\phi(\mathbf{r}).$$
\end{prop}
\begin{proof} 
By virtue of (\ref{conj}) , the identity (\ref{laplres}) and the assumptions of Proposition $$\left( G(\bar{z})^{*}[-\Delta\phi-z\phi ]\right)(x)=\phi(x,0,0)=0, \quad x\in[0,l].  $$ In accordance with (\ref{krein5}) this means that 
\begin{equation}\label{res7}
\left(R_{L}(z)[-\Delta\phi-z\phi]\right)(\mathbf{r})   = \left(R(z)[-\Delta\phi-z\phi]\right)(\mathbf{r})=\phi(\mathbf{r}).
\end{equation}
 Therefore $\phi\in\mathcal{D}(A_{L})\cap\mathcal{D}(A)$ and in view of (\ref{res7}) 
 \begin{equation*}\begin{array}{c} \left(A_{L}\phi\right)(\mathbf{r})=-\Delta\phi(\mathbf{r})
 -z\phi(\mathbf{r})+
	z\left(R_{L}(z)[-\Delta\phi-z\phi]\right)(\mathbf{r}) 
\\ =-\Delta\phi(\mathbf{r})-z\phi(\mathbf{r})+z\phi(\mathbf{r})=
 -\Delta\phi(\mathbf{r}). 
 \end{array}
 \end{equation*}
\end{proof}	
\begin{prop} Let $f(x_{1},x_{2},x_{3})$ be a function from $\mathcal{D}(A_{L})$ and $u_{f}(x), \, x\in[0,l],$ be defined by
\begin{equation}
u_{f}(x)=-\underset{\rho\rightarrow 0}{\lim}\frac{1}{\ln\left(  \rho^{2}\right) }f(x,x_{2},x_{3}), \quad \rho=\sqrt{x_{2}^{2}+x_{3}^{2}}. 
\end{equation}	
Then $u_{f}\in \mathcal{D}(L)$ and 
\begin{equation}\label{local2}\begin{array}{c}
\left( Lu_{f}\right)(x)=-4\pi\cdot\underset{\rho\rightarrow 0}{\lim}\left[f(x,x_{2},x_{3})+\ln\left( \frac{1}{\rho^{2}}\right)\cdot u_{f}(x)+2\ln{2}\cdot u_{f}(x) \right. \\ \left.-\int\limits_{0}^{l}\frac{s-x}{\left| s-x\right| }\ln\left| s-x\right|u_{f}^{\prime}\left(s \right)ds \right]. \end{array}
\end{equation}	
\end{prop}
\begin{proof}
Turning to the expressions (\ref{krein5}) and (\ref{conj}), we recall first of all that the functions from $\mathcal{D}(A)$ are continuous \cite{Kat}. Therefore for any $h(\mathbf{x})$ from $\mathbf{L}_{2}(\mathbf{R}_{3})$ the functions $\left( R(z)h\right) $ and $\left( G(z)^{*}h\right)(x)$ from $\mathbf{L}_{2}(\mathbf{R}_{3})$ and $\mathbf{L}_{2}([0,l])$, respectively are continuous. We also take into account that the domains of operators $L$ and $L+Q(z)$ coincide, since $Q(z)$ is a bounded operator. By our assumptions $"0"$ is a regular point of operator $L+Q(z), \, \mathrm{Im}z\neq 0.$ Therefore for any $h\in\mathbf{L}_{2}(\mathbf{R}_{3})$ the function
\begin{equation}\label{sob}
\hat{u}_{h}(x)=\left(\left[L+Q(z) \right]^{-1} G(z)^{*}h \right)(x) 
\end{equation}
belongs to $\mathcal{D}(L)$, that is to the Sobolev class $H^{2}_{2}([0,l])$ and satisfies the boundary conditions $\hat{u}_{h}(0)=\hat{u}_{h}(l)=0$. 

Writing any $f\in\mathcal{D}(A_{L})$ in the form $f(\mathbf{x})=\left( R_{L}(z)h\right)(\mathbf{x})$ with some $h\in\mathbf{L}_{2}(\mathbf{R}_{3})$ we can find the limiting value of $\left( R_{L}(z)h\right)(\mathbf{x}),$ when $\rho=\sqrt{x_{2}^{2}+x_{3}^{2}}\rightarrow 0$ and $x_{1}\in[0,l]$ using the following elementary assertion, the proof of which are left to the reader.
\begin{lem}\label{simp1}
Let $u(x)$ be continuously differentiable function on $[0,l]$ satisfying the conditions $u(0)=u(l)=0.$ Then
\begin{equation}\label{smp}\begin{array}{c}
\int\limits_{0}^{l}\frac{1}{\sqrt{(x-s)^{2}+\rho ^{2}}}u(s)ds\underset{\rho\rightarrow 0}{=} u(x)\cdot\ln\frac{1}{\rho^{2}} \\ +2\ln2u(x)-\int\limits_{0}^{l}\frac{s-x}{\left| s-x\right| }\ln\left| s-x\right|u^{\prime}\left(s \right)ds .
\end{array}
\end{equation}
\end{lem}
Using the expression (\ref{krein5}) for $\left( R_{L}(z)h\right)(\mathbf{x})$ and applying Lemma \ref{simp1} one can easily verify that
\begin{equation}
u_{f}(x)=-\underset{\rho\downarrow 0}{\lim}\frac{1}{\ln\left(  \rho^{2}\right) }f(x,x_{2},x_{3})=-\frac{1}{4\pi}\hat{u}_{h}
(x)\in\mathcal{D}(L)
\end{equation}
and
\begin{equation}\begin{array}{c}
\underset{\rho\downarrow 0}{\lim}\left[f(x,x_{2},x_{3})+\ln\left( \frac{1}{\rho^{2}}\right)\cdot u_{f}(x)+2\ln{2}\cdot u_{f}(x) \right. \\ \left.
-\int\limits_{0}^{l}\frac{s-x}{\left| s-x\right| }\ln\left| s-x\right|u_{f}^{\prime}\left(s \right)ds \right]=\left( G(z)^{*}h\right) \\-\left(Q(z)\left[L+Q(z) \right]^{-1} G(z)^{*}h \right)(x) =L\left(\left[L+Q(z) \right]^{-1} G(z)^{*}h \right)(x) \\
=L\hat{u}_{h}(x)=-\frac{1}{4\pi}\left( Lu_{f}\right)(x). 
\end{array}
\end{equation}
\end{proof}

\end{document}